\documentclass[]{article} 
 
\usepackage{a4}
\usepackage[utf8]{inputenc}
\usepackage{amssymb}
\usepackage{amsmath}
\usepackage{amsfonts}
\usepackage{amsthm}
\usepackage{moreverb,tabularx}
\usepackage{graphicx}
\usepackage{enumerate}
\usepackage{tabularx}
\usepackage{booktabs}
\usepackage{etex}
\usepackage{helvet}
\usepackage{psboxit,pstricks}
\usepackage[all]{xy}
\usepackage{tikz}
\usepackage{tikz-qtree}
\usetikzlibrary{trees,matrix,fit,shapes,calc}
\usepackage{array,multirow}
\usepackage[font=small,labelfont=bf]{caption}
\usepackage{subcaption}
\usepackage{mathtools}
\usepackage[colorlinks=false]{hyperref}
\usepackage{cleveref}[2012/02/15]
\crefformat{footnote}{#2\footnotemark[#1]#3}

\DeclarePairedDelimiter{\ceil}{\lceil}{\rceil}
\DeclarePairedDelimiter{\floor}{\lfloor}{\rfloor}

\newcommand{\TTT}{\mathcal{T}}

\newcommand{\s}[1]{\ensuremath{\mathtt{#1}}}
\newcommand{\zero}{\s{0}}
\newcommand{\one}{\s{1}}

\newcommand{\emptystring}{\varepsilon}
\newcommand{\concat}{\cdot}
\newcommand{\concats}{\cdots}

\newcommand{\catname}[1]{{\mathsf{#1}}}

\newcommand{\trees}{\catname{Trees}}

\newcommand{\caterpillars}{\catname{Caterpillars}}

\newcommand{\bintrees}{\catname{BinaryTrees}}

\newcommand{\nca}{\operatorname{nca}}

\newcommand{\lev}{\operatorname{lev}}
\newcommand{\ldepth}{\operatorname{ldepth}}
\newcommand{\parent}{\operatorname{parent}}

\newcommand{\depth}{\operatorname{depth}}

\newcommand{\size}{\operatorname{size}}
\newcommand{\lsize}{\operatorname{lsize}}
\newcommand{\apex}{\operatorname{apex}}
\newcommand{\llabel}{\operatorname{llabel}}
\newcommand{\hlabel}{\operatorname{hlabel}}

\newcommand{\hpath}{\operatorname{hpath}}
\newcommand{\hchild}{\operatorname{hchild}}
\newcommand{\lchildren}{\operatorname{lchildren}}

\newtheorem{theorem}{Theorem}[section]
\newtheorem{lemma}[theorem]{Lemma}
\newtheorem{corollary}[theorem]{Corollary}
\theoremstyle{definition}

\newtheorem*{theorem*}{Theorem}

\begin{document}

\title{\Large Near-optimal labeling schemes for nearest common ancestors
}
\author{Stephen Alstrup\thanks{Department of Computer Science, University of Copenhagen, Denmark, s.alstrup@diku.dk.} \\
\and 
Esben Bistrup Halvorsen\thanks{Department of Computer Science, University of Copenhagen, Denmark, esbenbh@diku.dk.}
\and 
Kasper Green Larsen\thanks{MADALGO - Center for Massive Data Algorithmics, a Center of the Danish National Research Foundation, Department of Computer Science, Aarhus University, Denmark, larsen@cs.au.dk.}
}
\date{}

\maketitle


\begin{abstract}

%
%

\small
We consider NCA labeling schemes: given a rooted tree $T$, label the nodes of $T$ with binary strings such that, given the labels of any two nodes, one can determine, by looking only at the labels, the label of their nearest common ancestor.

For trees with $n$ nodes we present upper and lower bounds establishing that labels of size $(2\pm \epsilon)\log n$, $\epsilon<1$ are both sufficient and necessary.\footnote{All logarithms in this paper are in base 2.}

Alstrup, Bille, and Rauhe (SIDMA'05) showed that ancestor and NCA labeling schemes have labels of size $\log n +\Omega(\log \log n)$. Our lower bound increases this to $\log n + \Omega(\log n)$ for NCA labeling schemes. Since Fraigniaud and Korman (STOC'10) established that labels in ancestor labeling schemes have size $\log n +\Theta(\log \log n)$, our new lower bound  separates ancestor and NCA labeling schemes. Our upper bound improves the $10 \log n$ upper bound by Alstrup, Gavoille, Kaplan and Rauhe (TOCS'04), and our theoretical result even outperforms some recent experimental studies by Fischer (ESA'09) where variants of the same NCA labeling scheme are shown to all have labels of size approximately $8 \log n$.

\end{abstract}


\section{Introduction}
\newcommand{\bitset}{\mbox{$\{0,1\}$}}
\newcommand{\Oh}{\mathrm{O}}
\newcommand{\ie}{i.e., }
\newcommand{\etal}{et al.}
\newcommand{\set}[1]{\left\{{#1}\right\}}

A \emph{labeling scheme} assigns a \emph{label}, which is a binary string, to each node of a tree such that, given only the labels of two nodes, one can compute some predefined function of the two nodes. The main objective is to minimize the \emph{maximum label length}: that is, the maximum number of bits used in a label. 

With labeling schemes it is possible, for instance, to avoid costly access to large, global tables, to compute locally in distributed settings, and to have storage used for names/labels be informative. These properties are used in XML search engines~\cite{AKM01}, network routing and distributed algorithms~\cite{TZ01,Gavoille01,Cowen01,EGP98a,Gavoille01,FG02}, graph representations~\cite{KNR92} and other areas. An extensive survey of labeling schemes can be found in~\cite{gavoillepeleg}.

A nearest common ancestor (NCA) labeling scheme labels the nodes such that, for any two nodes, their labels alone are sufficient to determine the label of their NCA. Labeling schemes can be found, for instance, for distance, ancestor, NCA, connectivity, parent and sibling~\cite{GPPR01,KKP00,Peleg99,AGKR04,KNR92,TZ01,AR02,KM01,Breuer66,BF67,siamcompKatzKKP04,SK85}, and have also been analyzed for dynamic trees~\cite{CohenKaplan2010}.  NCA labeling schemes are used, among other things, to compute minimum spanning trees in a distributed setting~\cite{Pagli2004,Flocchini2012,Blin2010}.

Our main result establishes that labels of size $(2\pm\epsilon)\log n$, $\epsilon<1$ are both necessary and sufficient for NCA labeling schemes for trees with $n$ nodes.
More precisely, we show that label sizes are lower bounded by $1.008\log n -O(1)$ and upper bounded by $2.772\log n+O(1)$.

Since our lower bound is $\log n+\Omega(\log n)$, this establishes an exponential separation (on the nontrivial, additive term) between NCA labeling and the closely related problem of ancestor labeling which can be solved optimally with labels of size $\log n +\Theta(\log \log n)$~\cite{fraigniaudkorman,alstrupbillerauhe}.
(An ancestor labeling scheme labels the nodes in a tree such that, for any two nodes, their labels alone are sufficient to determine whether the first node is an ancestor of the second.)
The upper bound of $\log n +\Oh(\log \log n)$ for ancestor~\cite{fraigniaudkorman} is the latest result in a sequence \cite{AKM01,KM01,KMS02,AR02,abiteboul,fraigniaudkorman2} of improvements from the trivial $2\log n$ bound~\cite{Tsakalidis84}.

Our upper bound improves the $10\log n$ label size of~\cite{AGKR04}. In addition to the NCA labeling scheme used to establish our upper bound, we present another scheme with labels of size $3 \log n$ which on the RAM uses only linear time for preprocessing and constant time for answering queries, meaning that it may be an efficient solution for large trees compared to traditional non-labeling scheme algorithms~\cite{HT84}.

NCAs, also known as \emph{least common ancestors} or \emph{lowest common ancestors} (LCAs), have been studied extensively over the last several decades in many variations; see, for example,~\cite{Maier77,ASSU81,AHU76,SV88,CH99,AT00,Powell90,BF00,GBT84,SV88,BV93}. A linear time algorithm to preprocess a tree such that subsequent NCA queries can be answered in constant time is described in~\cite{HT84}. NCAs have numerous applications for graphs~\cite{Gabow90,KKT95,DRT92,AHU76}, strings~\cite{Gusfield97,Farach97}, planarity testing~\cite{Westbrook92}, geometric problems~\cite{CN99,GBT84}, evolutionary trees~\cite{FKW95}, bounded tree-width algorithms~\cite{CZ98} and more. A survey on NCAs with variations and application can be found in~\cite{AGKR04}.

A $\log n + \Oh(\log^*n)$ adjacency labeling scheme is presented in~\cite{alstruprauhe}, and adjacency labeling schemes of $\log n + \Oh(1)$ are presented in~\cite{bonichon} for the special cases of binary trees and caterpillars. We present NCA labeling schemes with labels of size $2.585\log n +O(1)$ and $\log n + \log \log n+O(1)$ for binary trees and caterpillars, respectively. Our lower bound holds for any family of trees that includes all trees of height $O(\log n)$ in which all nodes either have 2 or 3 children. 

\subsection{Variations and related work.}
The NCA labeling scheme in~\cite{AGKR04} is presented as an $\Oh(\log n)$ result, but it is easy to see that the construction gives labels of worst-case size $10 \log n$. The algorithm uses a decomposition of the tree, where each  component is assigned a sub-label, and a label for a node is a combination of sub-labels. Fischer~\cite{fischer} ran a series of experiments using various techniques for sub-labels~\cite{AGKR04,mehlhorn,hutucker} and achieved experimentally that worst-case label sizes are approximately $8 \log n$.

Peleg~\cite{Peleg00a} has established labels of size $\Theta(\log^2 n)$ for NCA labeling schemes in which NCA queries have to return a \emph{predefined label} of $\Oh(\log n)$ bits. Experimental studies of this variation can be found in~\cite{CaminitiFP08}.
In~\cite{Blin2010} the results from~\cite{AGKR04} are extended to predefined labels of length $k$. We have included a corollary that shows that such an extension can be achieved by adding $k \log n$ bits to the labels.

In~\cite{KormanK07} a model (1-query) is studied where one, in addition to the label of the input nodes, can access the label of one additional node. With this extra information, using the result from~\cite{AGKR04} for NCA labeling, they present a series of results for NCA and distance.
As our approach improves the label length from~\cite{AGKR04}, we also improve some of the label lengths in~\cite{KormanK07}.

Sometimes various computability requirements are imposed on the labeling scheme: in~\cite{siamcompKatzKKP04} a query should be computable in polynomial time; in~\cite{AKM01} in constant time on the RAM; and in~\cite{KNR92} in polynomial time on a Turing machine.
We use the same approach as in~\cite{AGKR04}, but with a different kind of sub-labels and with different encodings for lists of strings, and it is only the $2.772 \log n+O(1)$ labeling scheme for trees that we do not show how to implement efficiently.


\section{Preliminaries}
The \emph{size} or \emph{length} of a binary string $s=s_1\concats s_k$ is the number of bits $|s|=k$ in it. The concatenation of two strings $s$ and $t$ is denoted $s\concat t$.

Let $T$ be a rooted tree with root $r$. The \emph{depth} of a node $v$, denoted $\depth(v)$, is the length of the unique path from  $r$ to $v$. If a node $u$ lies on the path from $r$ to a node $v$, then $u$ is an \emph{ancestor} of $v$ and $v$ is a \emph{descendant} of $u$. If, in addition, $\depth(v)=\depth(u)+1$ so that $uv$ is an edge in the tree, then $u$ is the unique \emph{parent} of $v$, denoted  $\parent(v)$, and $v$ is a \emph{child} of $u$. 
A \emph{binary tree} is a rooted tree in which any node has at most two children.
A \emph{common ancestor} of two nodes $v$ and $w$ is a node that is an ancestor of both $v$ and $w$, and their \emph{nearest common ancestor} (NCA), denoted $\nca(v,w)$, is the unique common ancestor with maximum depth. The descendants of $v$ form an induced subtree $T_v$ with $v$ as root. The \emph{size} of $v$, denoted $\size(v)$, is the number of nodes in $T_v$.

Let $\TTT$ be a family of rooted trees. 
An \emph{NCA labeling scheme for $\TTT$} consists of an \emph{encoder} and a \emph{decoder}. The encoder is an algorithm that accepts  any tree $T$ from $\TTT$ as input and produces a \emph{label} $l(v)$, which is a binary string, for every node $v$ in $T$. The decoder is an algorithm that takes two labels $l(v)$ and $l(w)$ as input and produces the label $l(\nca(v,w))$ as output. Note that encoder knows the entire tree when producing labels for nodes, whereas the decoder knows nothing about $v$, $w$ or the tree from which they come, although it does know that they come from the \emph{same} tree and that this tree belongs to $\TTT$.
The worst-case \emph{label size} is the maximum size of a label produced by the encoder from any node in any tree in $\TTT$.

\section{Lower bound}
This section introduces a class of integer sequences, 3-2 sequences, and an associated class of trees, 3-2 trees\footnote{The related ``2-3 trees''~\cite{aho2} have a slightly different definition, which is why we use a different terminology here.}, so that two 3-2 trees that have many labels in common when labeled with an NCA labeling scheme correspond to two 3-2 sequences that are ``close'' in the sense of a metric known as Levenshtein distance. By considering  a subset of 3-2 sequences that are pairwise distant in this metric, the corresponding set of 3-2 trees cannot have very many labels in common, which leads to a  lower bound on the total number of labels and hence on the worst-case label size.

\subsection{Levenshtein distance and 3-2 sequences.}
The \emph{Levenshtein distance}~\cite{levenshtein}, or \emph{edit distance}, between two sequences $x$ and $y$ is defined as the number $\lev(x,y)$ of single-character edits (insertion, deletion and substitution) required to transform $x$ into $y$. 
%
A \emph{3-2 sequence} of length $2k$ is an integer sequence  $x=(x_1,\dots ,x_{2k})$ with exactly $k$ 2s and $k$ 3s.
\begin{lemma} \label{lemm:32seqs}
For any $h,k$ with $2\leq h\leq k$ and $k$ an integer with $k\geq 90$, there exists a set $\Sigma$ of 3-2 sequences of length $2k$ with $|\Sigma|\geq 2^{1.95k}/(16k/h)^{3h}$ and $\lev(x,y)>h$ for all $x,y\in\Sigma$.
\end{lemma}
\begin{proof}
Since  $\lev(x,y)>h$ is equivalent to $\lev(x,y)>\floor{h}$ and $2^{1.95k}/(16k/h)^{3h}\leq 2^{1.95k}/(16k/\floor{h})^{3\floor{h}}$, we can safely assume that $h$ is an integer. 

Now, let $x$ be an arbitrary 3-2 sequence of length $2k$, and consider the number of 3-2 sequences $y$ of length $2k$ with $\lev(x,y)\leq h$. We can transform $x$ into $y$ by performing $r$ deletions followed by $s$ substitutions followed by $t$ insertions, where $r+s+t\leq h$. 
This leads to the following upper bound on the number of $y$'s:
\[
\sum_{r=0}^h {2k \choose r} \sum_{s=0}^{h-r} {2k-r \choose s}  \sum_{t=0}^{h-r-s} {2k-r-s+t \choose t}2^t 
 \leq (h+1)^3 {2k \choose h}^2{3k\choose h} 2^h.
 \]
 Using Stirling's approximation~\cite{robbins} and the fact that $(h+1)^3\leq 8^h$ for all $h\geq 2$, it follows that this is upper bounded by
\[  8^h (2ke/h)^{2h}(3ke/h)^h 2^h 
= 3^h(4ke/h)^{3h} 
\leq (16k/h)^{3h}.
\]

We now construct $\Sigma$ as follows. Let $\Sigma'$ denote the set of 3-2 sequences of length $2k$, and note that $|\Sigma'|= {2k \choose k}$. Pick an arbitrary 3-2 sequence $x$ from $\Sigma'$, add it to $\Sigma$ and remove all strings $y$ from $\Sigma'$ with $\lev(x,y)\leq h$. Continue by picking one of the remaining strings from $\Sigma'$, add it to $\Sigma$ and remove all strings  from $\Sigma'$ within distance $h$. When we run out of strings in $\Sigma'$ we will, according to the previous calculation and Stirling's approximation~\cite{robbins}, have
\[
|\Sigma|\geq \frac{{2k \choose k}}{(16k/h)^{3h}} \geq \frac{2^{2k-1}}{k^{1/2}(16k/h)^{3h}} \geq  \frac{2^{1.95k}}{(16k/h)^{3h}},
\]
where the last inequality follows from the fact that 
$2^{0.05k-1}\geq k^{1/2}$ whenever  $k\geq 90$.
\end{proof}

\subsection{3-2 trees and a lower bound.} \label{sec:32trees}
Given a 3-2 sequence $x=(x_1,\dots ,x_{2k})$ of length $2k$, we can create an associated tree of depth $2k$ where all nodes at depth $i-1$ have exactly $x_i$ children, and all nodes at depth $2k$ are leaves. We denote this tree the \emph{3-2 tree associated with $x$}. The number of nodes at depth $i$ in the 3-2 tree associated with $x$ is $x_1\cdots x_i$; in particular, the number of leaves is $x_1\cdots x_{2k}=6^k$. The number of nodes in total is upper bounded by $2\cdot 6^k$.

Consider the set of labels produced by an NCA labeling scheme for the nodes in a tree. Given a subset $S$ of these labels, let $S'$ denote the set of labels which can be generated from $S$ by the labeling scheme: thus, $S'$ contains the labels in $S$ as well as the labels for the NCAs of all pairs of nodes labeled with labels from $S$. The labels in $S'$ can be organized as a rooted tree according to their ancestry relations, which can be determined directly from the labels using the decoder of the labeling scheme and without consulting the original tree. The tree produced in this way is denoted $T^S$ and is uniquely determined from $S$.
Note that, if all the nodes in $S$ are leaves, then all internal nodes in $T^S$ must have been obtained as the NCA of two leaves, and hence must have at least two children.

Now, given a tree $T^S$ induced by a subset $S$ of labels assigned to the leaves of a tree $T$ by an NCA labeling scheme, we can create an integer sequence, $I(S)$, as follows. Start at the root of $T^S$, and let the first integer be the number of children of the root. Then recurse to a child $v$ for which the subtree $T_v^S$ contains a maximum number of leaves, and let the second integer be the number of children of this child. Continue this until a leaf is reached (without writing down the last $0$). Note that, if $T$ is a 3-2 tree of depth $2k$, the produced sequence $I(S)$ will have length at most $2k$ and will contain only 2s and 3s.

\begin{lemma} \label{lemm:32trees}
Let $T$ be a 3-2 tree associated with the 3-2 sequence $x=(x_1,\dots ,x_{2k})$. Let $S$ be a set of $m$ labels assigned to the leaves of $T$ by an NCA labeling scheme. Then $\lev(x,I(S))\leq \log_{3/2} (6^k/m)$.
\end{lemma}
\begin{proof}
We describe a way to transform $x$ into $I(S)$. Start at the root of $T$, and let $i$ be the depth in $T$ containing the node $v$ whose label $l(v)$ is the root in $T^S$. Delete all entries $x_1,\dots ,x_{i-1}$ from $x$ and compare the number of children of $l(v)$ in $T^S$ to $x_i$. If the numbers are the same, leave $x_i$ be; if not, we must have that $x_i=3$ and that the number of children of $l(v)$ is $2$, so replace $x_i$ by $2$. Then recurse to a child $w$ of $v$ in $T$ for which the corresponding subtree in $T^S$ contains a maximum number of leaves, and repeat the process with $T_w$, the corresponding subtree of $T^S$ and the remaining elements $x_{i+1},\dots ,x_k$.

Clearly, this transforms $x$ into $I(S)$ using only deletions and substitutions, where all substitutions replace a $3$ by a $2$. Each of these edits modify the maximum possible number of leaves in $T^S$ compared to $T$ with a factor of either $1/2$ or $2/3$. It follows that the number $m$ of leaves in $T^S$ satisfies $m\leq 6^k\cdot (2/3)^{\lev(x,I(S))}$, which implies $\lev(x,I(S))\leq\log_{3/2}(6^k/m)$ as desired.
\end{proof}

We now present our main lower bound result. The result is formulated for a family $\TTT$ that is large enough to contain all 3-2 trees with $N$ nodes; in particular, it holds for the family of all rooted trees with at most $N$ nodes. 

\begin{theorem}\label{theo:ncalowertrees}
If $\TTT$ is a family of trees that contains all 3-2 trees with up to $N\geq 2\cdot 3^{240}$ nodes, then any NCA labeling scheme for $\TTT$ has a worst-case label size of at least $1.008\log N - 318$.
\end{theorem}
\begin{proof}
Let $k=120\floor{\frac{1}{120}\log_6 (N/2)}$ be $\log_6 (N/2)$ rounded down to the nearest multiple of $120$, and let $n=6^k\leq N/2$. 
Further, set $m=n^{119/120}$ and $h=2\log_{3/2}(n/m)$. Note that $n$, $m$ and $n/m=n^{1/120}$ are all integers. Observe also that $n > (N/2)/6^{120}\geq (3/2)^{120}$ and thereby that $h\geq 2$. Finally, observe that $h=\frac{1}{60}k\log_{3/2}6\leq k$ and that $k\geq 120$.

According to \Cref{lemm:32seqs}, there exists a set $\Sigma$ of 3-2 sequences of length $2k$ with $|\Sigma| \geq 2^{1.95k}/(16k/h)^{3h}$ and $\lev(x,y)>h$ for all $x,y\in\Sigma$.
The set $\Sigma$ defines a set of $|\Sigma|$ associated 3-2 trees with $n$ leaves and at most $2 n\leq N$ nodes. In particular, all the associated trees belong to $\TTT$.
We can estimate the number of elements in $\Sigma$ as follows: 
\begin{align*}
|\Sigma| &\geq \frac{2^{1.95k}}{(16k/h)^{3h}} \\
& = \frac{2^{1.95\log_6 n}}{(8\log_6 n/\log_{3/2}(n/m))^{6\log_{3/2}(n/m)}} \\
& = \frac{n^{1.95\log_6 2}}{(960\log_6 n/\log_{3/2}n)^{(6\log_{3/2}n)/120}} \\
& = \frac{n^{1.95\log_6 2}}{(960\log_6 (3/2))^{0.05\log_{3/2}n}} \\
& = \frac{n^{1.95\log_6 2}}{n^{0.05\log_{3/2}(960\log_6(3/2))}} \\
& = n^{1.95\log_6 2 - 0.05\log_{3/2}(960\log_6(3/2))}  \\
& \geq n^{0.09}
\end{align*}

Now suppose that an NCA labeling scheme labels the nodes of all 3-2 trees associated with sequences in $\Sigma$. Consider two trees associated with sequences $x,y\in \Sigma$, and let $S$ denote the set of leaf labels that are common to $x$ and $y$. We must then have $|S|<m$, since otherwise, by \Cref{lemm:32trees}, we would have 
\[
\lev(x,y)\leq \lev(x,I(S))+\lev(I(S),y) \leq \frac{h}{2}+\frac{h}{2} = h.
\]
It follows that, if we restrict attention to a subset $\TTT$ consisting of $\min(|\Sigma|,\floor{n/(2m)})$ of the trees associated with strings in $\Sigma$, then the leaves of any tree in $\TTT$ can share a total of at most $n/2$ labels with all other trees in $\TTT$. In other words, every tree in $\TTT$ has at least $n/2$ leaf labels that are unique for this tree within the set of all leaf labels of trees in $\TTT$. This gives a total of at least
\begin{align*}
\frac{n}{2}\min(|\Sigma|,\floor{n/(2m)}) 
&= \frac{n}{2}\min(n^{0.09},\floor{n^{1/120}/2}) \\
&= n^{121/120}/8 \\
& \geq n^{1.008}/8
\end{align*}
distinct labels. If the worst-case label size is $L$, we can create $2^{L+1}-1$ distinct labels, and we must therefore have $n^{1.008}/8 \leq 2^{L+1}-1$ from which it follows that
\begin{align*}
L\geq \floor{1.008\log n} -3 &\geq \floor{1.008\log (N/2\cdot 6^{120})} -3 \\
& \geq 1.008\log N - 318. \qedhere
\end{align*}
\end{proof}

\section{Upper bound}
In this section we construct an NCA labeling scheme that assigns to every node a label consisting of a sequence of sub-labels, each of which is constructed from a decomposition of a tree known as heavy-light decomposition. The labeling scheme is similar to that of~\cite{AGKR04} but with a different way of constructing sub-labels (presented in \Cref{sec:onencascheme}), a different way of ordering sub-labels (presented in \Cref{sec:ordering}) and a different way of encoding lists of sub-labels (presented in \Cref{sec:encodings}). 
\subsection{Encodings.} \label{sec:encodings}
We begin with a collection of small results that show how to efficiently encode sequences of binary strings.
\begin{lemma} \label{lemm:labels}
A collection of $n$ objects can be uniquely labeled with binary strings of length at most $L$ if and only if $L\geq\floor{\log n}$.
\end{lemma}
\begin{proof}
There are $2^L$ binary strings of length $L$, and hence there are $2^{L+1}-1$ binary strings of length at most $L$. Thus, we can create unique labels for $n$ different objects using labels of length at most $L$ whenever $n\leq 2^{L+1}-1$, which is equivalent to $L\geq \ceil{\log(n+1)}-1 = \floor{\log n}$. (The latter equality follows from the simple fact that $\floor{r}=\ceil{s}-1$ for all real numbers $r<s$ for which there does not exist an integer $z$ with $r<z<s$.)
\end{proof}

\begin{lemma} \label{lemm:labelssamesize}
A collection of $n$ objects can be uniquely labeled with binary strings of length \emph{exactly} $L$ if and only if $L\geq\ceil{\log n}$.
\end{lemma}
\begin{proof}
The argument is similar to the one in \Cref{lemm:labels}, but with the modification that we only use labels of length \emph{exactly} equal to $L$. This yields the inequality $n\leq 2^L$, which is equivalent to $L\geq \ceil{\log n}$.
\end{proof}

\Cref{lemm:labels,lemm:labelssamesize} can only be efficiently implemented if there is a way to efficiently implement the 1-1 correspondence between the objects and the numbers $1,\dots ,n$. The remaining lemmas of this section show how to encode sequences of binary strings whose concatenation has length $t$, and all of them except \Cref{lemm:encodenonemptypairs2} can be implemented with linear time encoding and constant time decoding on a RAM machine in which a machine word has size $O(t)$.

\begin{lemma} \label{lemm:encodepairs}
Let $a=(a_1,a_2)$ be a pair of (possibly empty) binary strings with $|a_1\concat a_2|=t$. We can encode $a$ as a single binary string of length $t+\ceil{\log t}$ such that a decoder without any knowledge of $a$ or $t$ can recreate $a$ from the encoded string alone.
\end{lemma}
\begin{proof}
Since $|a_1|\leq |a_1\concat a_2|=t$, we can use \Cref{lemm:labelssamesize} to encode $|a_1|$ with exactly $\ceil{\log t}$ bits. We then encode $a$ by concatenating the encoding of $|a_1|$ with $a_1\concat a_2$ to give a string of exactly $t+\ceil{\log t}$ bits. Since $t$ is uniquely determined from $t+\ceil{\log t}$, the decoder can split up the encoded string into the encoding of $|a_1|$ and the concatenation $a_1\concat a_2$ from which it can recreate $a_1$ and $a_2$.
\end{proof}

We thank Mathias B{\ae}k Tejs Knudsen 
for inspiring parts of the proof of \Cref{lemm:encodenonemptypairs2} below. As the proof shows, the encoding in \Cref{lemm:encodenonemptypairs2} is optimal with respect to size but comes with no guarantees for time complexities. \Cref{lemm:encodenonemptypairs} further below is a suboptimal version of \Cref{lemm:encodenonemptypairs2} but with a more efficient implementation.
\begin{lemma} \label{lemm:encodenonemptypairs2}
Let $a=(a_0,\dots ,a_{2k})$ be a list of (possibly empty) binary strings with $|a_0\concats a_{2k}|=t$ and with $a_{2i}\concat a_{2i+1}\neq \emptystring$ for all $i<k$. We can encode $a$ as a single binary string of length $\ceil{(1+\log(2+\sqrt{2}))t}$ such that a decoder without any knowledge of $a$, $t$ or $k$ can recreate $a$ from the encoded string alone.
\end{lemma}
\begin{proof}
We will use \Cref{lemm:labelssamesize} to encode $a$ for a fixed $t$. To do this, we must count the number of possible sequences in the form of $a$. There are $2^t$ choices for the $t$ bits in the concatenation $a_0\concats a_{2k}$, and every subdivision of the concatenation into the substrings $a_i$ corresponds to a solution to the equation
\[
x_0+x_1+\dots +x_{2k}=t
\]
where $x_{2i}+x_{2i+1}\geq 1$ for $i=0,\dots ,k-1$. Note that we must have $k\leq t$.  For a given $t$, let $s_t$ denote the number of solutions (including choices of $k$) to the above equation. We shall prove further below that 
\begin{equation}\label{eq:lengthexpression}
s_t=\frac{1}{4}c^{t+1}+\frac{1}{4}d^{t+1},
\end{equation}
where $c=2+\sqrt{2}$ and $d=2-\sqrt{2}$, which easily implies $s_t\leq (2+\sqrt{2})^t$. It then follows that the total number of sequences $a$ for fixed $t$ is bounded by $2^t(2+\sqrt{2})^t$, and using \Cref{lemm:labelssamesize} we can therefore encode any such $a$ as a string with \emph{exactly} $\ceil{(1+\log(2+\sqrt{2}))t}$ bits. Since $t$ is uniquely determined by this length, the decoder can determine $t$ from the length of the string and then use \Cref{lemm:labelssamesize} to recreate $a$.

It remains to show~\eqref{eq:lengthexpression}. For any $t$, the number of solutions with $k=0$ is $1$. 
Given a solution where $k>0$, let $j=x_0+x_1$, and note that $j\geq 1$ and that $x_2+\cdots +x_{2k}=t-j$ is a solution to the problem for $t-j$. There are $j+1$ solutions to $x_0+x_1=j$, and hence the total number of solutions is
\[
s_t = 1+ \sum_{j=1}^{t} s_{t-j}(j+1).
\]
Using this expression, it is straightforward to see that
\[
s_t-2s_{t-1}+s_{t-2}=2s_{t-1}-s_{t-2},
\]
which implies $s_t = 4s_{t-1}-2s_{t-2}$. The characteristic polynomial of this recurrence relation has roots $c$ and $d$, and hence $s_t=\alpha c^t+\beta d^t$ for some $\alpha,\beta$. Using $s_0=1$ and $s_1=3$ to solve, we obtain $\alpha=c/4$ and $\beta=d/4$, which proves~\eqref{eq:lengthexpression}.
\end{proof}

\begin{lemma} \label{lemm:encodenonemptypairs}
Let $a=(a_0,\dots ,a_{2k})$ be a list of (possibly empty) binary strings with $|a_0\concats a_{2k}|=t$ and with $a_{2i}\concat a_{2i+1}\neq \emptystring$ for all $i<k$. We can encode $a$ as a single binary string of length $3t$ such that a decoder without any knowledge of $a$, $t$ or $k$ can recreate $a$ from the encoded string alone.
\end{lemma}
\begin{proof}
We encode $a$ as a concatenation of three binary strings of lengths $t$, $t-1$ and $t+1$, respectively. The first string is the concatenation $\tilde a =a_0\concats a_{2k}$. The second string has a $\one$ in the $i$'th position for $i\leq t-1$ exactly when the $(i+1)$'th position of $\tilde a$ is the first bit in a substring $a_{2j}\concat a_{2j+1}$ (which by the assumption is nonempty for all $j$). The third string has a $\one$ in the $i$'th position for $i\leq t$ exactly when the $i$'th position of $\tilde a $ is the first bit in a substring $a_{2j+1}$ for some $j$ or in $a_{2k}$, and a $\one$ in the $(t+1)$'th position exactly when $a_{2k}\neq \emptystring$.

If the decoder receives the concatenation of length $3t$ of these three strings, it can easily recreate the three strings by splitting up the string into three substrings of sizes $t$, $t-1$ and $t+1$. The first string is $\tilde a$, which it can then split up at all positions where the second string has a $\one$.  This gives a list of nonempty strings in the form $a_{2i}\concat a_{2i+1}$ for $i\leq k-2$ as well as the string $a_{2k-2}\concat a_{2k-1}\concat a_{2k}$. The decoder can then use the third string  to split up each of these concatenations as follows. For every (nonempty) concatenation $a_{2i}\concat a_{2i+1}$, consider the corresponding bits in the third string. If one of these bits is a $\one$, then the concatenation should be split up at that position; in particular, if the $\one$ is at the first bit in the concatenation, then it means that $a_{2i}$ is empty. If none of the bits is a $\one$, then it means that $a_{2i+1}$ is empty. In all cases, we can recreate $a_{2i}$ and $a_{2i+1}$. Likewise, the concatenation $a_{2k-2}\concat a_{2k-1}\concat a_{2k}$ can be split up using the $\one$s in the corresponding bits in the third string. If there are two $\one$s among these bits, then it is clear how to split up the concatenation. If there are no $\one$s, then it means that $a_{2k-1}$ and $a_{2k}$ are both empty. If there is exactly one $\one$, then we can split up the concatenation into $a_{2k-2}$ and $a_{2k-1}\concat a_{2k}$, and exactly one of $a_{2k-1}$ and $a_{2k}$ must be empty. The last bit of the third string determines which of these two cases we are in.
\end{proof}

\begin{lemma} \label{lemm:encodenonemptyconsecutives}
Let $a=(a_0,\dots ,a_k)$ be a list of (possibly empty) binary strings with $|a_0\concats a_k|=t$ and with $a_i\concat a_{i+1}\neq \emptystring$ for all $i< k$. We can encode $a$ as a single binary string of length $\ceil{(1+\log 3)(t-1)}+3$ such that a decoder without any knowledge of $a$, $t$ or $k$ can recreate $a$ from the encoded string alone.
\end{lemma}
\begin{proof}
We encode $a$ by concatenating $\tilde a=a_0\concats a_k$ of length $t$ with a string $s$ of length $\ceil{(\log 3) t}$. To describe $s$, we first construct a string $\tilde s$ of length $t-1$ over the alphabet $\{\zero,\one,\s{2}\}$. The $i$'th bit $\tilde s_i$ of $\tilde s$ is defined according to the role of the $(i+1)$'th bit $x$ in $\tilde a$ as follows:
\[
\tilde s_i = \begin{cases} 
\zero, & \text{if $x$ is the first bit of a nonempty string $a_j$,} \\
          & \text{where $a_{j-1}$ is nonempty,} \\
\one,  & \text{if $x$ is the first bit of a nonempty string $a_j$,} \\
          & \text{where $a_{j-1}$ is empty,} \\
\s{2}, & \text{else.} 
\end{cases}
\]
The string $\tilde s$ represents a unique choice out of $3^{t-1}$ possibilities, and by \Cref{lemm:labelssamesize} we can represent this choice with a binary string $s$ of length exactly equal to $\ceil{\log 3^{t-1}}=\ceil{(t-1)\log 3}$. We concatenate this with a single indicator bit representing whether $a_0$ is empty or not, and another indicator bit representing whether $a_k$ is empty or not. Finally, we concatenate all this with $\tilde a$, giving a string of total length $\ceil{(t-1)\log 3}+2+t = \ceil{(1+\log 3)(t-1)}+3$.

Since the value $t$ is uniquely determined from the length of the encoded string, the decoder is able to split up the encoded string into $\tilde a$, $s$ and the two indicator bits. It can then convert $s$ to $\tilde s$ and use the entries in $\tilde s$ and the indicator bits to recreate $a$ from $\tilde a$. This proves the theorem.
\end{proof}

\subsection{An order on binary strings.} \label{sec:ordering}
Consider the total order $\preceq$ on binary strings defined by 
\[
s\concat \zero\concat t\prec s\prec s\concat \one\concat t'
\]
for all binary strings $s,t,t'$. Here we have written $s\prec t$ as short for $s\preceq t\wedge s\neq t$. This order naturally arises in many contexts and has been studied before; see, for example,~\cite{TZ01}.
All binary strings of length three or less are ordered by $\preceq$ as follows:
{\small
\[
\zero\zero\zero\prec \zero\zero\prec \zero\zero\one \prec \zero \prec \zero\one\zero \prec \zero\one \prec \zero\one\one \prec \emptystring
\prec \one\zero\zero \prec \one\zero \prec \one\zero\one \prec \one \prec \one\one\zero \prec \one\one\prec \one\one\one
\]
}%

A finite sequence $(a_i)$ of binary strings is \emph{$\prec$-ordered} if $a_i\prec a_j$ for $i<j$. 

\begin{lemma}\label{lemm:halvorsen}
Given a finite sequence $(w_i)$ of positive numbers with $w=\sum_iw_i$, there exists an $\prec$-ordered sequence $(a_i)$ with $|a_i|\leq \floor{\log w - \log w_i}$ for all $i$.
\end{lemma}	
\begin{proof}
The proof is by induction on the number of elements in the sequence $(w_i)$. If there is only one element, $w_1$, then we can set $a_1=\emptystring$, which satisfies $|a_1|=0 =\floor{\log w_1-\log w_1}$. So suppose that there is more than one element in the sequence and that the theorem holds for shorter sequences. Let $k$ be the smallest index such that $\sum_{i\leq k} w_i>w/2$, and set $a_k=\emptystring$.  Then $a_k$ clearly satisfies the condition. The subsequences $(w_i)_{i<k}$ and $(w_i)_{i>k}$ are shorter and satisfy $\sum_{i<k}w_i\leq w/2$ and $\sum_{i>k}w_i\leq w/2$, so by induction there exist $\prec$-ordered sequences $(b_i)_{i<k}$ and $(b_i)_{i>k}$ with $|b_i|\leq \floor{\log(w/2)-\log w_i} = \floor{\log w-\log w_i} -1$ for all $i\neq k$. Now, define $a_i$ for $i<k$ by $a_i=0\concat b_i$ and for $i>k$ by $a_i=1\concat b_i$. Then $(a_i)$ is a $\prec$-ordered sequence with $|a_i|\leq\floor{\log w-\log w_i}$ for all $i$.
\end{proof}
A linear time implementation of the previous lemma can be achieved as follows. First compute the numbers $t_i=\floor{\log w-\log w_i}$ in linear time. Now set $a_1=\zero^{t_1}$ to be the minimum (with respect to the order $\preceq$) binary string of length at most $t_1$. At the $i$'th step, set  $a_i$ to be the minimum binary string of length at most $t_i$ with $a_{i-1}\prec a_i$. If this process successfully terminates, then the sequence $(a_i)$ has the desired property. On the other hand, the process must terminate, because the above lemma says that there \emph{exists} an assignment of the $a_i$'s, and our algorithm conservatively chooses each $a_i$ so that the set of possible choices left for $a_{i+1}$ is maximal at every step. A similar argument shows that the following lemma can be implemented in linear time.

\begin{lemma}\label{lemm:halvorsennonempty}
Given a finite sequence $(w_i)$ of positive numbers with $w=\sum_iw_i$, there exist an $\prec$-ordered sequence $(a_i)$ of nonempty strings and a $k$  such that $|a_i|\leq \floor{\log (w+w_k) - \log w_i}$ for all $i$.
\end{lemma}	
\begin{proof}
Let $k$ be the smallest index such that $\sum_{i\leq k} w_i>w/2$ and add an extra copy of $w_k$ next to $w_k$ in the sequence of weights. The total sequence of weights will now sum to $w+w_k$, and if we apply \Cref{lemm:halvorsen} to this sequence, exactly one of the two copies of $w_k$ will be assigned the empty string. Discard this string, and what is left is a $\prec$-ordered sequence $(a_i)$ with $|a_i|\leq \floor{\log(w+w_k)-\log w_i}$ for all $i$ as desired.
\end{proof}

\subsection{Heavy-light decomposition.}\label{sec:heavylightdecomposition}
We next  describe the \emph{heavy-light decomposition} of  Harel and Tarjan~\cite{HT84}. 
Let $T$ be a rooted tree. The nodes of $T$ are classified as either \emph{heavy} or \emph{light} as follows. The root $r$ of $T$ is light. For each internal node $v$, pick one child node $w$ where $\size(w)$ is maximal among the children of $v$ and classify it as heavy; classify the other children of $v$ as light. We denote the unique heavy child of $v$ by $\hchild(v)$ and the set of light children by $\lchildren(v)$. The \emph{light size} of a node $v$ is the number $\lsize(v)=1+\sum_{{w\in\lchildren(v)}}\size(w)$, which is equal to $\size(v)-\size(\hchild(v))$ when $v$ is internal.
The \emph{apex} of $v$, denoted $\apex(v)$, is the nearest light ancestor of $v$.
By removing the edges between light nodes and their parents, $T$ is divided into a collection of \emph{heavy paths}. The set of nodes on the same heavy path as  $v$ is denoted $\hpath(v)$. The top node of $\hpath(v)$ is the light node $\apex(v)$. 

For a node $v$, consider the sequence $u_0,\dots ,u_k$ of light nodes encountered on the path from the root $r=u_0$ to $v$. The number $k$ is the \emph{light depth} of $v$, denoted $\ldepth(v)$. The light depth of $T$, $\ldepth(T)$ is the maximum light depth among the nodes in $T$. 
Note that $\ldepth(v)\leq \ldepth(T)\leq \log n$; see~\cite{HT84}.

\subsection{One NCA labeling scheme.} \label{sec:onencascheme}
We now describe the labeling scheme that will be used for various  families of trees, although with different encodings for each family.
Given a rooted tree $T$, we begin by assigning to each node $v$ a \emph{heavy label}, $\hlabel(v)$, and, when $v$ is light and not equal to the root, a \emph{light label}, $\llabel(v)$, as described in \Cref{lemm:heavylabels,lemm:lightlabels} below.
\begin{lemma} \label{lemm:heavylabels}
There exist binary strings $\hlabel(v)$ for all nodes $v$ in $T$ so that the following hold for all nodes $v,w$ belonging to the same heavy path:
\begin{gather}
\begin{multlined}[t] \depth(v)<\depth(w) \implies \\ \hlabel(v) \prec \hlabel(w)  \label{eq:hlabel} \end{multlined} \\
|\hlabel(v)| \leq \floor{\log \size(\apex(v)) -\log \lsize(v)} \label{eq:hlabelsize}
\end{gather}
\end{lemma}
\begin{proof}
Consider  each heavy path $H$ separately and use the sequence $(\lsize(v))_{v\in H}$, ordered ascendingly by $\depth(v)$, as input to \Cref{lemm:halvorsen}.
\end{proof}

\begin{lemma} \label{lemm:lightlabels}
There exist binary strings $\llabel(v)$ for all light nodes $v\neq r$ in $T$ so that the following hold for all light siblings $v,w$:
\begin{gather}
v\neq w\implies \llabel(v) \neq \llabel(w) \label{eq:llabel} \\
|\llabel(v)| \leq  \floor{\log \lsize(\parent(v)) - \log \size(v)} \label{eq:llabelsize} 
\end{gather}
\end{lemma}
\begin{proof}
Consider each set $L$ of light siblings separately and use the sequence $(\size(v))_{v\in L}$, not caring about order, as input to  \Cref{lemm:halvorsen}.
\end{proof}
In many cases we are not going to use the constructions in \Cref{lemm:heavylabels,lemm:lightlabels} directly, but will instead use the following two modifications:

\begin{lemma} \label{lemm:lightlabelsnonempty}
It is possible to modify the constructions in \Cref{lemm:heavylabels,lemm:lightlabels} so that, for all nodes $u,v$ where $v$ is a light child of $u$,
\begin{equation} \label{eq:llabelnonempty}
\hlabel(u) =\emptystring \implies \llabel(v)\neq \emptystring .
\end{equation}
The modification still satisfies~\eqref{eq:hlabel},~\eqref{eq:hlabelsize},~\eqref{eq:llabel} and~\eqref{eq:llabelsize} except that when $\hlabel(u)$ is empty, \eqref{eq:llabelsize} is replaced by
\begin{equation} \label{eq:hllabelsize}
|\hlabel(u)| + |\llabel(v)| \leq  
\floor{\log \size(\apex(u)) - \log \size(v)} 
\end{equation}
\end{lemma}
\begin{proof}
First observe that without modifying the construction in \Cref{lemm:heavylabels,lemm:lightlabels} we can combine~\eqref{eq:hlabelsize} with~\eqref{eq:llabelsize} to obtain~\eqref{eq:hllabelsize}.
We now describe the modification: the construction works exactly as in the two lemmas except that in cases where $\hlabel(u)$ is empty, we use \Cref{lemm:halvorsennonempty} in place of \Cref{lemm:halvorsen} in the construction of light labels in \Cref{lemm:lightlabels}. This clearly makes~\eqref{eq:llabelnonempty} true, so it remains to prove~\eqref{eq:hllabelsize}.

So let $u$ and $v$ be as above. By construction of the heavy-light decomposition,  $\size(\hchild(u))$ is larger than or equal to the size of any of the light children of $u$, and hence larger than the size that corresponds to the weight $w_k$ in \Cref{lemm:halvorsennonempty}. Further,  $\lsize(u)+\size(\hchild(u))=\size(u)\leq \size(\apex(u))$. Using these two facts together, \Cref{lemm:halvorsennonempty} now yields
\begin{equation*}
|\llabel(v)| \leq \floor{\log \size(\apex(u)) - \log \size(v)}.
\end{equation*}
Since $|\hlabel(u)|=0$, we have therefore obtained~\eqref{eq:hllabelsize}.
\end{proof}

\begin{lemma} \label{lemm:heavylabelsnonempty}
It is possible to modify the constructions in \Cref{lemm:heavylabels,lemm:lightlabels} so that, for all nodes $u,v,w$ where $v$ is a light child of $u$ and $w$ is a descendant of $v$ on the same heavy path as $v$,
\begin{equation} \label{eq:hlabelnonempty}
\hlabel(u) =\emptystring \text{ and } \llabel(v) =\emptystring \implies \hlabel(w)\neq \emptystring .
\end{equation}
The modification still satisfies~\eqref{eq:hlabel},~\eqref{eq:hlabelsize},~\eqref{eq:llabel} and~\eqref{eq:llabelsize} except that when $\hlabel(u)$ and $\llabel(v)$ are both empty,~\eqref{eq:hlabelsize} is replaced by
\begin{equation} \label{eq:hlhlabelsize}
|\hlabel(u)| + |\llabel(v)| + |\hlabel(w)|\leq 
\floor{\log \size(\apex(u)) - \log \lsize(w)} 
\end{equation}
\end{lemma}
\begin{proof}
The proof is similar to that of the previous lemma. First observe that without modifying the construction in \Cref{lemm:heavylabels,lemm:lightlabels} we can combine~\eqref{eq:hlabelsize},~\eqref{eq:llabelsize} and~\eqref{eq:hlabelsize} again to obtain~\eqref{eq:hlhlabelsize}.
We now describe the modification: the construction works exactly as in the two lemmas except that in cases where $\hlabel(u)$ and $\llabel(v)$ are both empty, we use \Cref{lemm:halvorsennonempty} in place of \Cref{lemm:halvorsen} in the construction of heavy labels in \Cref{lemm:heavylabels}. This clearly makes~\eqref{eq:hlabelnonempty} true, so it remains to prove~\eqref{eq:hlhlabelsize}.

So let $u$, $v$ and $w$ be as above. Note that $\size(v)$ is larger than or equal to the light size of any of the nodes on the heavy path with $v$ as apex, and hence larger than the light size that corresponds to the weight $w_k$ in \Cref{lemm:halvorsennonempty}. Further,  $2\size(v)\leq \lsize(u)+\size(\hchild(u))=\size(u))\leq \size(\apex(u))$. Using these two facts together, \Cref{lemm:halvorsennonempty} now yields
\begin{equation*}
|\hlabel(w)| \leq \floor{\log \size(\apex(u)) - \log \lsize(w)}.
\end{equation*}
Since $|\hlabel(u)|=|\llabel(v)|=0$, we have therefore obtained~\eqref{eq:hlhlabelsize}.
\end{proof}

We next assign a new set of labels for the nodes of $T$. Given a node $v$ with $\ldepth(v)=k$, consider the sequence $u_0,v_0,\dots ,u_k,v_k$ of nodes from the root $r=u_0$ to $v=v_k$, where $u_i=\apex(v_i)$ is light for $i=0,\dots ,k$ and $v_{i-1}=\parent(u_i)$ for $i=1,\dots ,k$. Let $l(v)=(h_0,l_1,h_1,\dots ,l_k,h_k)$, where $l_i=\llabel(u_i)$ and $h_i=\hlabel(v_i)$.
\Cref{fig:ncalabels} shows an example of a tree with the labels $l(v)$. 
Note that we have used \Cref{lemm:heavylabels,lemm:lightlabels} for the construction of labels in this figure and not any of the modifications in \Cref{lemm:lightlabelsnonempty,lemm:heavylabelsnonempty}.

\begin{figure*}
\centering
\begin{tikzpicture}
\Tree [
.{\footnotesize{\underline{$\emptystring$}}}		
	\edge[dashed]; [.{\footnotesize{\underline{$\emptystring$}$\emptystring$\underline{\zero}}} [.{\footnotesize{\underline{$\emptystring$}$\emptystring$\underline{$\emptystring$}}} \edge[dashed]; [.{\footnotesize{\underline{$\emptystring$}$\emptystring$\underline{$\emptystring$}\one\underline{\zero}}} {\footnotesize{\underline{$\emptystring$}$\emptystring$\underline{$\emptystring$}\one\underline{$\emptystring$}}} ]
							                       					     [.{\footnotesize{\underline{$\emptystring$}$\emptystring$\underline{\one}}} {\footnotesize{\underline{$\emptystring$}$\emptystring$\underline{\one\one}}} \edge[dashed]; {\footnotesize{\underline{$\emptystring$}$\emptystring$\underline{\one}$\emptystring$\underline{$\emptystring$}}} ]
							                       					     \edge[dashed]; [.{\footnotesize{\underline{$\emptystring$}$\emptystring$\underline{$\emptystring$}\zero\underline{\zero}}} {\footnotesize{\underline{$\emptystring$}$\emptystring$\underline{$\emptystring$}\zero\underline{$\emptystring$}}} ]
					    ] 
					    \edge[dashed]; {\footnotesize{\underline{$\emptystring$}$\emptystring$\underline{\zero}$\emptystring$\underline{$\emptystring$}}} ]
	\edge[dashed];  [.{\footnotesize{\underline{$\emptystring$}\zero\underline{\zero}}} {\footnotesize{\underline{$\emptystring$}\zero\underline{$\emptystring$}}} ]
	[.{\footnotesize{\underline{\one\zero}}} \edge[dashed]; {\footnotesize{\underline{\one\zero}$\emptystring$\underline{$\emptystring$}}}
					[.{\footnotesize{\underline{\one}}} \edge[dashed]; {\footnotesize{\underline{\one}$\emptystring$\underline{$\emptystring$}}}
									[.{\footnotesize{\underline{\one\one}}} \edge[dashed]; {\footnotesize{\underline{\one\one}$\emptystring$\underline{$\emptystring$}}}
							                				 {\footnotesize{\underline{\one\one\one}}}
							                				 \edge[dashed]; {\footnotesize{\underline{\one\one}\zero\underline{$\emptystring$}}} ]
									\edge[dashed]; {\footnotesize{\underline{\one}\zero\underline{$\emptystring$}}} ]
					\edge[dashed]; {\footnotesize{\underline{\one\zero}\zero\underline{$\emptystring$}}} ]
	]
\end{tikzpicture}
\caption{A tree with the labels $l(v)$ from \Cref{sec:onencascheme} and with heavy sub-labels underlined.} \label{fig:ncalabels}
\end{figure*}
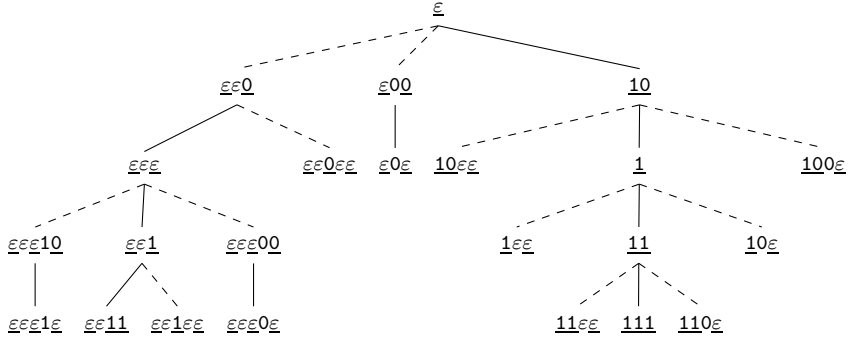

To define a labeling scheme, it remains to encode the lists $l(v)$ of binary strings into a single binary string. Before we do this, however, we note that $l(\nca(v,w))$ can be computed directly from $l(v)$ and $l(w)$. The proof is essentially the same as that in~\cite{AGKR04} although with the order $\preceq$ in place of the usual lexicographic order.

\begin{lemma} \label{lemm:ncaworks}
Let $v$ and $w$ be nodes in $T$, and let $u=\nca(v,w)$.
\begin{enumerate}[(a)]
\item \label{xprefixy} If $l(v)$ is a prefix of $l(w)$, then $l(u)=l(v)$.
\item \label{yprefixx} If $l(w)$ is a prefix of $l(v)$, then $l(u)=l(w)$.
\item \label{distinctl} If $l(v)=(h_0, l_1,\dots , h_i, l_i,\dots)$ and $l(w)=(h_0,l_1,\dots , h_i, l_i',\dots)$ with $l_i\neq l_i'$, then $l(u)=(h_0,l_1,h_1,\dots ,h_i)$.
\item \label{distincth} If $l(v)=(h_0, l_1,\dots , l_{i-1},h_i,\dots )$ and $l(w)=(h_0,l_1,h_1,\dots ,l_{i-1}, h_i',\dots)$ with $h_i\neq h_i'$, then $l(u)=(h_0,l_1,h_1,\dots  l_{i-1}, \min_{\preceq}\{h_i,h_i'\})$.
\end{enumerate}
\end{lemma}
\begin{proof}
By construction, $l=l(\parent(\apex(u)))$ is a prefix of both $l(v)$ and $l(w)$, and
\[
l(u)=l\concat (\llabel(\apex(u)), \hlabel(u)).
\]

Suppose first that $v$ is an ancestor of $w$, so that $u=v$, and let $x$ be the nearest ancestor of $w$ on $\hpath(v)$. Then $\apex(x)=\apex(u)$, so 
\[
l(w)=l\concat (\llabel(\apex(u)), \hlabel(x),\dots)
\]
If $u=x$, then $\hlabel(x)=\hlabel(u)$ and case~\eqref{xprefixy} applies. (If $v=w$ then case~\eqref{yprefixx} applies too.) If $u\neq x$, then $\hlabel(u)\prec\hlabel(x)$ by~\eqref{eq:hlabel} and case~\eqref{distincth} applies.
The case where $w$ is an ancestor of $v$ is analogous.

Suppose next that $v$ and $w$ are not ancestors of each other. Then $u$ must have children $\hat v$ and $\hat w$ with $\hat v\neq \hat w$  such that $\hat v$ is an ancestor of $v$ and $\hat w$ is an ancestor of $w$. At most one of $\hat v$ and $\hat w$ can be heavy. If neither of them are heavy, then they are apexes for their own heavy paths, and hence
\[
l(v)=l(u)\concat (\llabel(\hat v),\dots )
\]
and
\[
l(w)=l(u)\concat (\llabel(\hat w),\dots).
\]
By~\eqref{eq:llabel}, $\llabel(\hat v)$ and $\llabel(\hat w)$ are distinct, so case~\eqref{distinctl} applies. If $\hat v$ is heavy, then $\apex(\hat v)=\apex(u)$ and $l(v)=l\concat (\llabel(\apex(u)),\hlabel(\hat v),\dots)$ while $l(w)$ is still on the above form, i.e.\ $l(w)=l\concat (\llabel(\apex(u)),\hlabel(u),\dots)$. By~\eqref{eq:hlabel}, $\hlabel(u)\prec\hlabel(\hat v)$, so \eqref{distincth} applies. The case where $\hat w$ is heavy is analogous.
\end{proof}
Note that, as in~\cite{AGKR04}, the above theorem can be used to find labels for NCAs in constant time on the RAM as long as the labels have size $O(\log n)$.

As a final step, before presenting the encodings of the labels $l(v)$, we present a lemma that makes it easier to compute the size of the encodings. For brevity, we let $\tilde l(v) = h_0\concat l_1\concat h_1\concats l_k\concat h_k$ denote the concatenation of the sub-labels of $l(v)$.
\begin{lemma} \label{lemm:concatsize}
If $T$ has $n$ nodes, then $|\tilde l(v)|\leq \floor{\log n}$ for every node $v$ in $T$.
This holds no matter if we use \Cref{lemm:heavylabels,lemm:lightlabels} combined or any of the variants in \Cref{lemm:lightlabelsnonempty,lemm:heavylabelsnonempty} for the construction of heavy and light labels.
\end{lemma}
\begin{proof}
Let $v$ be an arbitrary node in $T$ and recall that $l(v)=(h_0,l_1,h_1,\dots ,l_k,h_k)$ where $l_i=\llabel(u_i)$ and $h_{i}=\hlabel(v_i)$ for nodes  $u_i,v_i$, $i=0,\dots ,k$ given by $r=u_0$, $v=v_k$,  $u_{i}=\apex(v_i)$ for all $i=0,\dots ,k$ and $v_{i-1}=\parent(u_i)$ for $i=1,\dots ,k$. 
If we use \Cref{lemm:heavylabels,lemm:lightlabels} for the construction of heavy and light labels, we have by~\eqref{eq:hlabelsize} that $|h_i|\leq\floor{\log\size(u_i)-\log \lsize(v_i)}$ for all $i=0,\dots ,k$ and by~\eqref{eq:llabelsize} that $|l_i|\leq  \floor{\log \lsize(v_{i-1}) - \log \size u_i}$ for $i=1,\dots ,k$. Summarizing now gives a telescoping sum:
\begin{align*}
|\tilde l(v)| & = |h_0\concat l_1\concat h_1 \concats l_k\concat h_k| \\
& \leq  \floor{\log \size(u_{0}) - \log \lsize (v_0)} + \\
& \quad\qquad \floor{\log \lsize (v_0) - \log\size(u_1)} +  \\
& \quad\qquad \cdots +  \floor{\log \size(u_{k}) - \log \lsize (v_k)} \\
& \leq \floor{\log \size(u_0)-\log \lsize(v_k)} \\
& \leq \floor{\log n}.
\end{align*}
In the cases where we have used any of the variants in~\Cref{lemm:lightlabelsnonempty,lemm:heavylabelsnonempty}, we must use~\eqref{eq:hllabelsize} or~\eqref{eq:hlhlabelsize} first to collapse sums of two or three terms in the above sum before collapsing the whole expression. Nevertheless, the result of the computation remains unchanged.
\end{proof}

\subsection{NCA labeling schemes for different families of trees.}
Let $\trees$ and $\bintrees$ denote the families of rooted trees and binary trees, respectively.

\begin{theorem} \label{theo:ncauppertrees}
There exists an NCA labeling scheme for $\trees$
whose worst-case label size is at most $\ceil{(1+\log(2+\sqrt{2}))\floor{\log n}} \leq 2.772\log n + 1$.
\end{theorem}
\begin{proof}
The encoder uses the modified construction in \Cref{lemm:lightlabelsnonempty} to ensure that every empty heavy label is followed by a nonempty light label. This means that the sequence $l(v)=(h_0,l_1,h_1,\dots ,l_k,h_k)$ can be encoded using  $\ceil{(1+\log(2+\sqrt{2}))\floor{\log n}}$ bits; see~\Cref{lemm:encodenonemptypairs2}.
Given the encoded labels from two nodes, the decoder can now decode the labels as described in \Cref{lemm:encodenonemptypairs2}, use \Cref{lemm:ncaworks} to compute the label of the NCA, and then re-encode that label using \Cref{lemm:encodenonemptypairs2} once again.
\end{proof}
The labeling scheme in \Cref{theo:ncauppertrees} makes use of \Cref{lemm:encodenonemptypairs2} which comes without any guarantees for the time complexities for encoding and decoding. This makes the result less applicable in practice. \Cref{theo:ncauppertrees2,theo:ncaupperbintrees,theo:ncauppercaterpillars} and \Cref{coro:predefinedlabels} below all use linear time for encoding and constant time for decoding.
\begin{theorem} \label{theo:ncauppertrees2}
There exists an NCA labeling scheme for $\trees$
whose worst-case label size is at most $3\floor{\log n}$.
\end{theorem}
\begin{proof} 
The proof is identical to  that of \Cref{theo:ncauppertrees} but with \Cref{lemm:encodenonemptypairs} in place of \Cref{lemm:encodenonemptypairs2}.
\end{proof}

A variant of NCA labeling schemes~\cite{Blin2010} allows every node to also have a \emph{predefined} label and requires the labeling scheme to return the predefined label of the NCA.
\begin{corollary} \label{coro:predefinedlabels}
There exists an NCA labeling scheme for $\trees$ with predefined labels of fixed length $k$ whose worst-case label size is at most $(3+k)\floor{\log n} +1$.
\end{corollary}
\begin{proof}
It suffices to save together with the NCA label of a node $v$ a table of the predefined labels for the at most $\floor{\log n}$ parents of light nodes on the path from the root to $v$, since the NCA of two nodes will always be a such for one of the nodes. By prepending a string in the form $\zero^i\one$ to the NCA label of $v$ we can ensure that it has size \emph{exacly} $3\floor{\log n}+1$. We can then append a table of up to $\floor{\log n}$  predefined labels of size $k$. Finally, we append $\zero$s to make the label have size exactly $(3+k)\floor{\log n} +1$. The decoder can now use the label's length to split up the label into the NCA label and the entries in the table of predefined labels.
\end{proof}

\begin{theorem} \label{theo:ncaupperbintrees}
There exists an NCA labeling scheme for $\bintrees$
whose worst-case label size is at most $\ceil{(1+\log 3)(\floor{\log n}-1)}+3 \leq  2.585\log n + 2$.
\end{theorem}
\begin{proof}
First note that every node in a binary tree has at most one light child. 
We can therefore assume that all light labels are empty.
Letting the encoder use the construction in \Cref{lemm:heavylabelsnonempty}, we can then ensure that every empty heavy label is followed by (an empty light label and) a nonempty heavy label. Since we can ignore light labels, it suffices to encode the sequence $(h_0,h_1,\dots ,h_k)$, and this sequence can be encoded with $\ceil{(1+\log 3)(\floor{\log n}-1)}+3$ bits; see \Cref{lemm:encodenonemptyconsecutives}. The rest of the proof follows the same argument as the proof of \Cref{theo:ncauppertrees}.
\end{proof}

A \emph{caterpillar} is a tree in which all leaves are connected to a single \emph{main path}. We assume caterpillars to always be rooted at one of the end nodes of the main path. Let $\caterpillars$ denote the family of caterpillars.
\begin{theorem} \label{theo:ncauppercaterpillars}
There exists an NCA labeling scheme for $\caterpillars$
whose worst-case label size is at most $\floor{\log n}+\ceil{\log\floor{\log n}}+1$.
\end{theorem}
\begin{proof}
By definition of caterpillars, every label $l(v)$ is either in the form $(h_0)$ or $(h_0,l_1,\emptystring)$. We encode the first case as $\zero\concat h_0$ and the second case as $\one\concat x$, where $x$ is the encoding of the pair $(h_0,l_1)$ using $\floor{\log n}+\ceil{\log\floor{\log n}}$ bits; see \Cref{lemm:encodepairs}. 
In both cases, the label size is at most $\floor{\log n}+\ceil{\log\floor{\log n}}+1$, and
the decoder can easily distinguish the two cases from the first bit. The rest of the proof follows the same argument as the proof of \Cref{theo:ncauppertrees}.
%
\end{proof}
For comparison, the best known lower bound for NCA labeling schemes for caterpillars is the trivial $\floor{\log n}$.

\bibliography{literature}

\providecommand{\bysame}{\leavevmode\hbox to3em{\hrulefill}\thinspace}
\providecommand{\MR}{\relax\ifhmode\unskip\space\fi MR }
\providecommand{\MRhref}[2]{%
  \href{http://www.ams.org/mathscinet-getitem?mr=#1}{#2}
}
\providecommand{\href}[2]{#2}
\begin{thebibliography}{10}

\bibitem{abiteboul}
S.~Abiteboul, S.~Alstrup, H.~Kaplan, T.~Milo, and T.~Rauhe, \emph{Compact
  labeling scheme for ancestor queries}, SIAM J. Comput. \textbf{35} (2006),
  no.~6, 1295--1309.

\bibitem{AKM01}
S.~Abiteboul, H.~Kaplan, and T.~Milo, \emph{Compact labeling schemes for
  ancestor queries}, Proceedings of the twelfth annual ACM-SIAM Symposium on
  Discrete Algorithms (SODA), 2001, pp.~547--556.

\bibitem{ASSU81}
A.~V. Aho, Y.~Sagiv, T.~G. Szymanski, and J.~D. Ullman, \emph{Inferring a tree
  from lowest common ancestors with an application to the optimization of
  relational expressions}, SIAM Journal on Computing \textbf{10} (1981), no.~3,
  405--421.

\bibitem{aho2}
Alfred~V. Aho, John~E. Hopcroft, and Jeffrey~D. Ullman, \emph{The design and
  analysis of computer algorithms}, Addison-Wesley, 1974.

\bibitem{AHU76}
A.V. Aho, J.E. Hopcroft, and J.D. Ullman, \emph{On finding lowest common
  ancestor in trees}, SIAM Journal on computing \textbf{5} (1976), no.~1,
  115--132, See also STOC 1973.

\bibitem{alstrupbillerauhe}
S.~Alstrup, P.~Bille, and T.~Rauhe, \emph{Labeling schemes for small distances
  in trees}, SIAM J. Discrete Math. \textbf{19} (2005), no.~2, 448--462.

\bibitem{AGKR04}
S.~Alstrup, C.~Gavoille, H.~Kaplan, and T.~Rauhe, \emph{Nearest common
  ancestors: A survey and a new algorithm for a distributed environment},
  Theory of Computing Systems \textbf{37} (2004), no.~3, 441--456.

\bibitem{AR02}
S.~Alstrup and T.~Rauhe, \emph{Improved labeling schemes for ancestor queries},
  Proc. of the 13th annual ACM-SIAM Symp. on Discrete Algorithms (SODA), 2002.

\bibitem{alstruprauhe}
\bysame, \emph{Small induced-universal graphs and compact implicit graph
  representations}, In Proc. 43rd annual IEEE Symp. on Foundations of Computer
  Science, 2002, pp.~53--62.

\bibitem{AT00}
S.~Alstrup and M.~Thorup, \emph{Optimal pointer algorithms for finding nearest
  common ancestors in dynamic trees}, Journal of Algorithms \textbf{35} (2000),
  no.~2, 169--188.

\bibitem{BF00}
M.~A. Bender and M.~Farach-Colton, \emph{The lca problem revisted}, 4th LATIN,
  2000, pp.~88--94.

\bibitem{BV93}
O.~Berkman and U.~Vishkin, \emph{Recursive star-tree parallel data structure},
  SIAM Journal on Computing \textbf{22} (1993), no.~2, 221--242.

\bibitem{Blin2010}
L.~Blin, S.~Dolev, M.~Potop-Butucaru, and S.~Rovedakis, \emph{Fast
  self-stabilizing minimum spanning tree construction: using compact nearest
  common ancestor labeling scheme}, Proceedings of the 24th international
  conference on Distributed computing, DISC'10, 2010, pp.~480--494.

\bibitem{bonichon}
N.~Bonichon, C.~Gavoille, and A.~Labourel, \emph{Short labels by traversal and
  jumping}, Electronic Notes in Discrete Mathematics \textbf{28} (2007),
  153--160.

\bibitem{Breuer66}
M.~A. Breuer, \emph{Coding vertexes of a graph}, IEEE Trans. on Information
  Theory \textbf{IT--12} (1966), 148--153.

\bibitem{BF67}
M.~A. Breuer and J.~Folkman, \emph{An unexpected result on coding vertices of a
  graph}, J. of Mathemathical analysis and applications \textbf{20} (1967),
  583--600.

\bibitem{CaminitiFP08}
S.~Caminiti, I.~Finocchi, and R.~Petreschi, \emph{Engineering tree labeling
  schemes: A case study on least common ancestors.}, ESA, Lecture Notes in
  Computer Science, vol. 5193, Springer, 2008, pp.~234--245.

\bibitem{CN99}
S.~Carlsson and B.~J. Nilsson, \emph{Computing vision points in polygons},
  Algorithmica \textbf{24} (1999), no.~1, 50--75.

\bibitem{CZ98}
S.~Chaudhuri and C.~D. Zaroliagis, \emph{Shortest paths in digraphs of small
  treewdith. {Part~II}: Optimal parallel algorithms}, Theoretical Computer
  Science \textbf{203} (1998), no.~2, 205--223.

\bibitem{CohenKaplan2010}
E.~Cohen, H.~Kaplan, and T.~Milo, \emph{Labeling dynamic xml trees}, SIAM J.
  Comput. \textbf{39} (2010), no.~5, 2048--2074.

\bibitem{CH99}
R.~Cole and R.~Hariharan, \emph{Dynamic lca queries on trees}, Annual ACM-SIAM
  Symposium on Discrete Algorithms (SODA), vol.~10, 1999.

\bibitem{Cowen01}
L.~J. Cowen, \emph{Compact routing with minimum stretch}, Journal of Algorithms
  \textbf{38} (2001), 170--183.

\bibitem{DRT92}
B.~Dixon, M.~Rauch, and R.~E. Tarjan, \emph{Verification and sensitivity
  analysis of minimum spanning trees in linear time}, SIAM Journal on Computing
  \textbf{21} (1992), no.~6, 1184--1192.

\bibitem{EGP98a}
T.~Eilam, C.~Gavoille, and D.~Peleg, \emph{Compact routing schemes with low
  stretch factor}, $17^{th}$ Annual ACM Symposium on Principles of Distributed
  Computing (PODC), August 1998, pp.~11--20.

\bibitem{Farach97}
M.~Farach-Colton, \emph{Optimal suffix tree construction with large alphabets},
  38th Annual Symposium on Foundations of Computer Science ({IEEE}, ed.), IEEE
  Computer Society Press, 1997, pp.~137--143.

\bibitem{FKW95}
M.~Farach-Colton, S.~Kannan, and T.~Warnow, \emph{A robust model for finding
  optimal evolutionary trees.}, Algorithmica \textbf{13} (1995), no.~1/2,
  155--179.

\bibitem{fischer}
J.~Fischer, \emph{Short labels for lowest common ancestors in trees}, ESA,
  2009, pp.~752--763.

\bibitem{Flocchini2012}
P.~Flocchini, T.~Mesa Enriquez, L.~Pagli, G.~Prencipe, and N.~Santoro,
  \emph{Distributed minimum spanning tree maintenance for transient node
  failures}, IEEE Trans. Comput. \textbf{61} (2012), no.~3, 408--414.

\bibitem{Gavoille01}
P.~Fraigniaud and C.~Gavoille, \emph{Routing in trees}, $28^{th}$ International
  Colloquium on Automata, Languages and Programming (ICALP), vol. 2076 of LNCS,
  2001, pp.~757--772.

\bibitem{FG02}
P.~Fraigniaud and C.~Gavoille., \emph{A space lower bound for routing in
  trees}, $19^{th}$ Annual Symposium on Theoretical Aspects of Computer Science
  (STACS), March 2002, pp.~65--75.

\bibitem{fraigniaudkorman2}
P.~Fraigniaud and A.~Korman, \emph{Compact ancestry labeling schemes for xml
  trees}, SODA, 2010, pp.~458--466.

\bibitem{fraigniaudkorman}
\bysame, \emph{An optimal ancestry scheme and small universal posets},
  Proceedings of the 42nd ACM symposium on Theory of computing (New York, NY,
  USA), 2010, pp.~611--620.

\bibitem{GBT84}
H.~N. Gabow, J.~L. Bentley, and R.~E. Tarjan, \emph{Scaling and related
  techniques for geometry problems}, Proc.~of the Sixteenth Annual {ACM}
  Symposium on Theory of Computing, 1984, pp.~135--143.

\bibitem{Gabow90}
H.N. Gabow, \emph{Data structure for weighted matching and nearest common
  ancestors with linking}, Annual ACM-SIAM Symposium on discrete algorithms
  (SODA), vol.~1, 1990, pp.~434--443.

\bibitem{gavoillepeleg}
C.~Gavoille and D.~Peleg, \emph{Compact and localized distributed data
  structures}, Distributed Computing \textbf{16} (2003), no.~2-3, 111--120.

\bibitem{GPPR01}
C.~Gavoille, D.~Peleg, S.~Perennes, and R.~Raz, \emph{Distance labeling in
  graphs}, 12th Symp. On Discrete algorithms, 2001.

\bibitem{Gusfield97}
D.~Gusfield, \emph{Algorithms on strings, trees, and sequences}, Cambridge
  University Press, 1997, pp. 196-207.

\bibitem{HT84}
D.~Harel and R.~E. Tarjan, \emph{Fast algorithms for finding nearest common
  ancestors}, Siam J. Comput \textbf{13} (1984), no.~2, 338--355.

\bibitem{hutucker}
T.~C. Hu and A.~C. Tucker, \emph{Optimum computer search trees}, SIAM Journal
  of Applied Mathematics \textbf{21} (1971), 514--532.

\bibitem{KNR92}
S.~Kannan, M.~Naor, and S.~Rudich, \emph{Implicit representation of graphs},
  SIAM J. DISC. MATH. (1992), 596--603, Preliminary version appeared in
  STOC'88.

\bibitem{KM01}
H.~Kaplan and T.~Milo, \emph{Short and distances and other functions}, 7nd
  Work. on Algo. and Data Struc., LNCS, 2001.

\bibitem{KMS02}
H.~Kaplan, T.~Milo, and R.~Shabo, \emph{A comparison of labeling schemes for
  ancestor queries}, Proceedings of the thirteen annual ACM-SIAM Symposium on
  Discrete Algorithms (SODA), 2002.

\bibitem{KKT95}
D.~R. Karger, P.~N. Klein, and R.~E. Tarjan, \emph{A randomized linear-time
  algorithm to find minimum spanning trees}, Journal of the ACM \textbf{42}
  (1995), no.~2, 321--328.

\bibitem{KKP00}
M.~Katz, N.~Katz, and D.~Peleg, \emph{Distance labeling schemes for
  well-seperated graph classes}, STACS'00, LNCS, vol. 1170, Springer Verlag,
  2000.

\bibitem{siamcompKatzKKP04}
M.~Katz, N.~A. Katz, A.~Korman, and D.~Peleg, \emph{Labeling schemes for flow
  and connectivity}, SIAM J. Comput. \textbf{34} (2004), no.~1, 23--40.

\bibitem{KormanK07}
A.~Korman and S.~Kutten, \emph{Labeling schemes with queries.}, SIROCCO, 2007,
  pp.~109--123.

\bibitem{levenshtein}
V.~I. Levenshtein, \emph{{Binary codes capable of correcting deletions,
  insertions and reversals.}}, Soviet Physics Doklady. \textbf{10} (1966),
  no.~8, 707--710.

\bibitem{Maier77}
D.~Maier, \emph{A space efficient method for the lowest common ancestor problem
  and an application to finding negative cycles}, 18th Annual Symposium on
  Foundations of Computer Science, 1977, pp.~132--141.

\bibitem{mehlhorn}
K.~Mehlhorn, \emph{A best possible bound for the weighted path length of binary
  search trees}, SIAM J. Comput. \textbf{6} (1977), no.~2, 235--239.

\bibitem{Pagli2004}
L.~Pagli, G.~Prencipe, and T.~Zuva, \emph{Distributed computation for swapping
  a failing edge}, Proceedings of the 6th international conference on
  Distributed Computing (Berlin, Heidelberg), IWDC'04, Springer-Verlag, 2004,
  pp.~28--39.

\bibitem{Peleg99}
D.~Peleg, \emph{Proximity-preserving labeling schemes and their applications},
  Graph-Theoretic concepts in computer science, 25th international workshop
  WG'99, LNCS, vol. 1665, Springer Verlag, 1999, pp.~30--41.

\bibitem{Peleg00a}
\bysame, \emph{Informative labeling schemes for graphs}, $25^{th}$
  International Symposium on Mathematical Foundations of Computer Science
  (MFCS), vol. 1893 of LNCS, Springer, August 2000, pp.~579--588.

\bibitem{Powell90}
P.~Powel, \emph{A further improved lca algorithm}, Tech. Report TR90-01,
  University of Minneapolis, 1990.

\bibitem{robbins}
H.~Robbins, \emph{A remark on {S}tirling's formula}, Amer. Math. Monthly
  \textbf{62} (1955), 26--29. \MR{MR0069328 (16,1020e)}

\bibitem{SK85}
N.~Santoro and R.~Khatib, \emph{Labeling and implicit routing in networks}, The
  computer J. \textbf{28} (1985), 5--8.

\bibitem{SV88}
B.~Schieber and U.~Vishkin, \emph{On finding lowest common ancestors:
  Simplification and parallelization}, SIAM Journal of Computing \textbf{17}
  (1988), 1253--1262.

\bibitem{TZ01}
M.~Thorup and U.~Zwick, \emph{Compact routing schemes}, ACM Symposium on
  Parallel Algorithms and Architectures, vol.~13, 2001.

\bibitem{Tsakalidis84}
A.~K. Tsakalidis, \emph{Maintaining order in a generalized linked list}, Acta
  Informatica \textbf{21} (1984), no.~1, 101--112.

\bibitem{Westbrook92}
J.~Westbrook, \emph{Fast incremental planarity testing}, Automata, Languages
  and Programming, 19th International Colloquium (Werner Kuich, ed.), Lecture
  Notes in Computer Science, vol. 623, Springer-Verlag, 1992, pp.~342--353.

\end{thebibliography}
\bibliographystyle{amsplain} 
\end{document}